\documentclass[journal,twoside,web]{ieeecolor}
\usepackage{generic}

\usepackage{amssymb}
\usepackage{mathrsfs}
\usepackage{algorithm}
\usepackage{algpseudocode}
\usepackage[noadjust]{cite}
\usepackage{bm}
\usepackage{nicefrac}
\usepackage{amsmath}
\usepackage{commath}
\usepackage[utf8]{inputenc}
\usepackage{enumerate}
\usepackage[hang,flushmargin]{footmisc}
\usepackage{tikz}
\usepackage{tabularx}
\usepackage{arydshln}
\usetikzlibrary{arrows}
\usetikzlibrary{shapes.geometric, shapes.multipart}
\tikzstyle{block}=[rectangle,text centered]
\usepackage{mathtools}

\usepackage{amsthm}
\newtheorem{thrm}{Theorem}

\newtheorem{remark}{Remark}
\newtheorem{lemma}{Lemma}
\newtheorem{assumpt}{Assumption}
\newtheorem{corollary}{Corollary}
\theoremstyle{definition}
\title{\LARGE \bf
A System Level Approach to Regret Optimal Control
}

\author{Alexandre Didier, Jerome Sieber and Melanie N. Zeilinger
\thanks{A. Didier, J. Sieber and M. N. Zeilinger are members of the Institute for Dynamic Systems and Control (IDSC), ETH Zurich, 8092 Zurich, Switzerland (e-mail: \{adidier, jsieber, mzeilinger\}@ethz.ch)}
}

\begin{document}

\maketitle
\thispagestyle{empty}
\pagestyle{empty}

\begin{abstract}
We present an optimisation-based method for synthesising a dynamic regret optimal controller for linear systems with potentially adversarial disturbances and known or adversarial initial conditions. The dynamic regret is defined as the difference between the true incurred cost of the system and the cost which could have optimally been achieved under any input sequence having full knowledge of all future disturbances for a given disturbance energy. This problem formulation can be seen as an alternative to classical $\bm{\mathcal{H}_2}$- or $\bm{\mathcal{H}_\infty}$-control. The proposed controller synthesis is based on the system level parametrisation, which allows reformulating the dynamic regret problem as a semi-definite problem. This yields a new framework that allows to consider structured dynamic regret problems, which have not yet been considered in the literature. For known pointwise ellipsoidal bounds on the disturbance, we show that the dynamic regret bound can be improved compared to using only a bounded energy assumption and that the optimal dynamic regret bound differs by at most a factor of $\nicefrac{\bm{2}}{\bm{\pi}}$ from the computed solution. Furthermore, the proposed framework allows guaranteeing state and input constraint satisfaction. 
\end{abstract}
\begin{IEEEkeywords}
Predictive control for linear systems, 
constrained control, 
optimal control, 
robust control.
\end{IEEEkeywords}
\section{INTRODUCTION} 
\IEEEPARstart{I}{n} classical control, the controller objective is commonly selected to minimise a given cost function. The two most common formulations are $\mathcal{H}_2$-control, for which the expected value of the cost is minimised under the assumption of stochastic noise, or $\mathcal{H}_\infty$-control, which minimises the worst-case cost given a potentially adversarial disturbance, see e.g. \cite{hassibi1999indefinite} for an overview. While $\mathcal{H}_2$-control can often be overly optimistic, $\mathcal{H}_\infty$-control often suffers from being overly conservative. 

Recently, there has been a growing interest in a more adaptive notion of a cost metric, i.e. regret minimisation. The idea behind this metric is simple. For any disturbance experienced by a system, a certain cost will be incurred even if the disturbances are known in advance. Therefore, the performance of a control algorithm is compared to the cost which could have been achieved, had all disturbances been known in advance. If an upper bound on this regret can be found, then the worst-case increase of the cost, which is incurred from not knowing all disturbances, is bounded. 

A popular metric in regret minimisation is policy regret. Here, the optimal benchmark cost is obtained by restricting the non-causal policy, which has access to all disturbances, to a certain class, e.g. static linear stabilising feedback controllers. 
The goal of policy regret minimising algorithms is to learn the optimal static feedback policy over a single episode of time length $T$, e.g., via gradient descent methods on disturbance feedback controllers.
This setting gives rise to interesting results for unknown, adversarial convex costs such as e.g. sublinear regret with respect to $T$ in \cite{agarwal2019online}, i.e. the average incurred regret goes to 0 as $T$ goes to infinity and the optimal non-causal controller in the restricted set is learned using computationally efficient algorithms. 
However, the provided regret bounds are often conservative with respect to the system parameters. 
Another popular notion for regret minimisation is dynamic regret. Here, the benchmark, which algorithms are compared to, no longer restricts the control actions to lie in a policy class, but is achieved through the optimal non-causal control sequence. For linear dynamics and quadratic costs, it was shown in \cite{goel2021regret} for a finite horizon and in \cite{sabag2021regret} for an infinite horizon, that the dynamic regret bound is proportional to the experienced disturbance energy. 
The derived regret optimal controller adapts to the experienced disturbances and can thus outperform the $\mathcal{H}_2$- and $\mathcal{H}_\infty$-controller. 

\textit{Contributions.} In Section~\ref{sec:RegOptSLS}, we present an optimisation-based method to synthesize a controller that minimises dynamic regret with respect to a given or an adversarial initial condition. The dynamic regret optimal controller is synthesised via a convex semi-definite program (SDP) rather than the bisection method in \cite{goel2021regret}. 
This SDP formulation, based on the system level parametrisation (SLP) introduced in \cite{anderson2019system}, enables to leverage pointwise ellipsoidal bounds on the disturbance to improve the regret bound, rather than only a bound on the disturbance energy used in the literature, as shown in Section~\ref{sec:PWB}. 
Additionally, a lower bound on the optimal dynamic regret for adversarial pointwise ellipsoidally bounded disturbances is provided. Finally, in Section~\ref{sec:Constr}, we show how the formulation can be extended to systems with state and input constraints.

\textit{Related Work.} A state-space formulation for the dynamic regret optimal controller with respect to bounded energy disturbances was derived for zero initial condition, for a finite horizon in \cite{goel2021regret} and for an infinite horizon in \cite{sabag2021regret}. The finite horizon solution can be found through bisection of a suboptimal performance level, similarly to $\mathcal{H}_\infty$-control. Concurrently to the work presented in this letter, \cite{martin2022safe} 
formulated an optimisation-based regret optimal synthesis only for disturbances with bounded energy, which allows including state and input constraints but does not consider the effect of pointwise ellipsoidally bounded disturbances nor a given initial state. 
In \cite{nonhoff2021online}, a dynamic regret bound is computed for constrained systems which are not subject to disturbances with unknown, strongly convex cost functions. 
The policy regret with respect to the stable stationary policy which robustly satisfies constraints on the system is minimised in \cite{li2021safe}, providing a sublinear in time regret bound. Finally, in \cite{anderson2019system} the synthesis of $\mathcal{H}_2$- and $\mathcal{H}_\infty$-controllers using the SLP is shown, which is closely related to the proposed synthesis. 

\textit{Notation.} A positive semi-definite matrix $A\in\mathbb{R}^{n\times n}$ fulfills $A\succeq 0$ and $A\succeq B$ implies that $A-B\succeq 0$ for some matrix $B\in\mathbb{R}^{n\times n}$. A block diagonal matrix $A\in\mathbb{R}^{nm\times nm}$ containing the matrices $B_i\in\mathbb{R}^{m\times m}$ for all $i\in\{1,{\dots}, n\}$ on its diagonal is denoted $A=\textup{blkdiag}(B_1,{\dots}, B_n)$. For a given matrix $A\in\mathbb{R}^{n\times n}$, the entry in the $i$-th row and $j$-th column is denoted as $[A]_{i,j}$, the $i$-th row and column of the matrix are given as $[A]_{i,:}$ and $[A]_{:,i}$, respectively, and the block matrix of appropriate dimensions at position $(i,j)$ as $[[A]]_{i,j}$ starting at index 1.
For a set $\mathcal{A}$, we denote as $\mathcal{A}^i$ the $i$-time Cartesian product $\mathcal{A}{\times} {\dots} {\times} \mathcal{A}$. For a vector $x\in\mathbb{R}^n$, $\norm{x}$ denotes the Euclidean norm and for a matrix $A\in\mathbb{R}^{n\times n}$, the induced $2$-norm $\norm{A}=\max_{\;\norm{x}{\neq} 0}\frac{\norm{Ax}}{\norm{x}}=\max_{\;\norm{x}=1}\norm{Ax}$.
The identity matrix and the vector of ones of appropriate dimensions are denoted as $\mathbb{I}$ and $\mathbf{1}$, respectively.
\section{PROBLEM FORMULATION}\label{sec:PF} 
We consider linear time-varying systems with discrete-time dynamics with an unknown additive disturbance of the form
\begin{equation}\label{eq:sysdyn}
x_{k+1}=A_k x_k+B_k u_k+E_k w_k,
\end{equation}
with states $x_k\in\mathbb{R}^n$, inputs $u_k\in\mathbb{R}^m$, unknown disturbance $w_k\in\mathbb{R}^r$ and known time-varying system matrices $A_k\in\mathbb{R}^{n\times n}$, $B_k\in\mathbb{R}^{n\times m}$ and $E_k \in\mathbb{R}^{n\times r}$.
At every time step, a known, potentially time-varying quadratic stage cost is incurred, given as
\begin{equation*}
l_k(x_k,u_k) = x_k^\top Q_kx_k + u_k^\top R_ku_k,
\end{equation*}
where the matrices $Q_k\in\mathbb{R}^{n \times n}$ and $R_k\in\mathbb{R}^{m\times m}$ are positive definite, i.e., $Q_k\succ0$ and $R_k\succ 0$.
The system is controlled over a single episode of time length $T$ and the corresponding sequences of states, inputs and disturbances are denoted as
		$\mathbf{x} = [ x_0^\top, x_1^\top, \dots, x_{T}^\top ]^\top,$ 
		$\mathbf{u} = [ u_0^\top, u_1^\top, \dots, u_{T}^\top ]^\top$ and 
		$\mathbf{w} = [ w_0^\top, w_1^\top, \dots, w_{T-1}^\top ]^\top$, respectively.
Additionally, we introduce the block diagonal cost matrices $\mathcal{Q}\in\mathbb{R}^{n(T+1)\times n(T+1)}$ and $\mathcal{R}\in\mathbb{R}^{m(T+1)\times m(T+1)}$, where
$\mathcal{Q}=\textup{blkdiag}(Q_0, \dots Q_{T})$ and $
\mathcal{R}=\textup{blkdiag}(R_0, \dots R_{T}).$
The total incurred cost during the episode is given by
\begin{equation}
J(x_0, \mathbf{u}, \mathbf{w})= \sum_{k=0}^{T} l_k(x_k,u_k) = \mathbf{x}^\top \mathcal{Q} \mathbf{x} + \mathbf{u}^\top \mathcal{R}\mathbf{u}.
\end{equation}
The objective is to minimise the dynamic regret incurred by the system, i.e. the difference to the best possible cost that could have been achieved through knowledge of all past and future disturbances. Note that minimising the dynamic regret can be seen as an alternative to classical approaches like $\mathcal{H}_2$-control, where the expected value of the incurred cost is minimised, and $\mathcal{H}_\infty$-control, where the worst-case cost is minimised. The optimal dynamic regret given an adversarial disturbance is
\begin{equation}\label{eq:DynReg1}
\textup{Regret}^* = \min_\mathbf{u} \max_{\norm{\mathbf{w}}^2=\omega}( J(x_0,\mathbf{u},\mathbf{w}) - J^*(x_0,\mathbf{w})),
\end{equation}
with $\omega\in\mathbb{R}$ and where the optimal non-causal cost is
\begin{equation}
J^*(x_0,\mathbf{w}) = \min_{\mathbf{\tilde{u}}(\mathbf{w})}  \mathbf{x}^\top \mathcal{Q} \mathbf{x} + \mathbf{\tilde{u}}^\top \mathcal{R}\mathbf{\tilde{u}}.
\end{equation} 
In Section~\ref{sec:OptNCCont}, a closed form expression for $J^*(x_0,\mathbf{w})$ and the corresponding control sequence $\mathbf{\tilde{u}}^*$ are derived. 
The dynamic regret incurred by system \eqref{eq:sysdyn} using the control sequence $\mathbf{u}^*$ is then given as
\begin{equation*}
\textup{Regret} = J(x_0,\mathbf{u}^*,\mathbf{w}) - J^*(x_0,\mathbf{w})
\end{equation*}
Note that the only assumption on $\mathbf{w}$ in the dynamic regret minimisation is that the disturbance is an $\ell_2$-signal. 
\section{REGRET OPTIMAL CONTROL VIA SYSTEM LEVEL PARAMETRISATION}\label{sec:RegOptSLS}
The problem formulated in Section~\ref{sec:PF} has been shown to have an optimal solution, with a state-space controller formulation given in \cite{goel2021regret} for zero initial condition. In \cite{goel2021regret}, the problem is converted into a $\mathcal{H}_\infty$-problem, for which the optimal solution is known. 
Similar to $\mathcal{H}_\infty$-synthesis, the controller can be computed via bisection on the performance level by imposing a linear matrix inequality (LMI) condition, see e.g., \cite{hassibi1999indefinite}. 
In the following, the solution to the optimal non-causal controller and its associated cost is detailed in Section~\ref{sec:OptNCCont} before a novel regret optimal controller synthesis via the solution of a convex SDP through the use of the SLP, see e.g. \cite{anderson2019system} for an overview, is proposed in Sections~\ref{sec:SLP} and \ref{sec:ROC}. A dynamic regret bound, which depends on the true disturbance energy which the system is exposed to, is provided, covering the case where $\norm{\mathbf{w}}^2$ is entirely unknown and the connections to $\mathcal{H}_\infty$-synthesis using the SLP are also shown in Section~\ref{sec:ROC}. The presented synthesis method allows incorporating information about pointwise bounds on the disturbance $\mathbf{w}$ in order to find a solution with a lower bound on the dynamic regret than using only the $\ell_2$-signal assumption as well as incorporating safety constraints into the synthesis, as shown in Section~\ref{sec:Ext}.
\subsection{Optimal non-causal controller} \label{sec:OptNCCont}
In order to synthesise the regret optimal controller, a closed-form solution for the optimal non-causal cost $J^*(x_0,\mathbf{w})$ is required. Note that the optimal non-causal controller and cost was derived e.g. in \cite{goel2021regret} and \cite{goel2020power} for zero initial condition and is provided below for completeness. 
By defining 
$\bm{\delta}= [ x_0^\top, \mathbf{w}^\top ]^\top,$ the state sequence can be obtained as
$\mathbf{x}=F\mathbf{u}+G\bm{\delta},$
where $F\in\mathbb{R}^{n(T+1) \times m(T+1)}$ and $G\in\mathbb{R}^{n(T+1) \times n(T+1)}$,
resulting in the optimal non-causal cost $$J^*(\bm{\delta}){=}\min_{\mathbf{\tilde{u}}} \bm{\delta}^\top G^\top \mathcal{Q}G\bm{\delta}+2\mathbf{\tilde{u}}^\top F^\top\mathcal{Q}G\bm{\delta}+\mathbf{\tilde{u}}^\top(\mathcal{R}+F^\top\mathcal{Q}F)\mathbf{\tilde{u}}.$$
The optimal non-causal control sequence can be computed as
\begin{equation}
\mathbf{\tilde{u}}^* = -(\mathcal{R}+F^\top\mathcal{Q} F)^{-1} F^\top\mathcal{Q} G\bm{\delta}
\end{equation}
and the corresponding optimal non-causal cost is given by
\begin{equation}
J^*(\bm{\delta}) = \bm{\delta}^\top G^\top (\mathcal{Q}^{-1}+F\mathcal{R}^{-1}F^\top)^{-1} G \bm{\delta} =: \bm{\delta}^\top O \bm{\delta} ,
\end{equation}
where the Woodbury matrix identity was used. Note that the optimal control sequence $\mathbf{\tilde{u}}^*$ is a linear combination of all disturbances as well as the initial state even though no structural assumptions were made. Additionally, $O$ depends only on the known system matrices $A_k, B_k$ and $E_k$ as well as the known cost functions $Q_k$ and $R_k$ and can therefore be used in Section~\ref{sec:ROC} in order to synthesise the regret optimal controller. 

\subsection{System Level Parametrisation}\label{sec:SLP}
In order to synthesise the regret optimal controller through an SDP, we make use of the SLP. This formulation allows to synthesise time-varying controllers for a wide variety of problem formulations, see e.g. \cite{anderson2019system}, and we show how it can be used in dynamic regret minimisation. We first define
$\mathcal{A} = \textup{blkdiag}(A_0, \dots, A_{T-1}, 0),$ $\mathcal{B} = \textup{blkdiag}(B_0,\dots,B_{T-1}, 0),$
and $\mathcal{E}=\textup{blkdiag}(\mathbb{I}, E_0,\dots,E_{T-2}, E_{T-1}),$ where $\mathcal{E}$ needs to be left invertible,
as well as the strictly causal, linear time-varying controller matrix 
\begin{equation*}
\mathcal{K} = \begin{bmatrix} K^{0,0}& 0 & \dots & 0 \\ K^{1,1} & K^{1,0} & \ddots & \vdots \\ \vdots & \ddots & \ddots & 0 \\ K^{T,T} & \dots & K^{T, 1} & K^{T,0} \end{bmatrix}, 
\end{equation*}
resulting in the SLP
\begin{equation}\label{eq:SLP}
\begin{split}
\mathbf{x} &= (\mathbb{I}-\mathcal{Z}(\mathcal{A}-\mathcal{B}\mathcal{K}))^{-1} \bm{\delta} =: \Phi_x \bm{\delta}, \\
\mathbf{u} &= \mathcal{K}\mathbf{x}=\mathcal{K}(\mathbb{I}-\mathcal{Z}(\mathcal{A}-\mathcal{B}\mathcal{K}))^{-1} \bm{\delta} =: \Phi_u \bm{\delta},
\end{split}
\end{equation} 
where $\mathcal{Z}$ is the block-downshift operator containing identity matrices on the first sub-diagonal and $\Phi_x, \Phi_u$ are causal operators, i.e., they are lower block diagonal.
The system response $\Phi$ is then defined as follows
\begin{equation} \label{eq:SLPrc}
\begin{bmatrix} \Phi_x \\ \Phi_u\end{bmatrix}=:\Phi=:[\Phi^0 \; \; \Phi^w]
\end{equation}
such that $\Phi$ is the stacked matrix of $\Phi_x$ and $\Phi_u$ and $\Phi_0\in\mathbb{R}^{(n+m)(T+1)\times n}$ and $\Phi_w\in\mathbb{R}^{(n+m)(T+1)\times nT}$ denote its respective column submatrices.
It then holds for all possible system responses $\Phi$, resulting from the dynamics \eqref{eq:sysdyn} that 
\begin{equation}\label{eq:SLPConstraint}
\begin{bmatrix} \mathbb{I}-\mathcal{Z}\mathcal{A} & -\mathcal{Z}\mathcal{B} \end{bmatrix} \Phi = \mathcal{E}
\end{equation}
and the controller matrix can be recovered as $\mathcal{K}=\Phi_u \Phi_x^{-1}$, which is shown in \cite[Theorem 2.1]{anderson2019system} and \cite[Theorem 1]{chen2021system}. Note that the regret optimal controller as shown in \cite{goel2021regret} consists of the optimal $\mathcal{H}_2$-controller and a linear combination of past disturbances, 
where the definitions of $M_k^{[i-1]}$ are omitted for brevity.
As shown in \cite[Theorem 2]{sieber2021system}, the class of controllers of the SLP is equivalent to the class of disturbance feedback parametrisations, which implies that the regret optimal controller can be synthesised through the SLP.

\subsection{Regret Optimal Control}\label{sec:ROC}
In this section, we derive the SDP formulation to synthesise the regret optimal controller. 
By substituting \eqref{eq:SLP} into \eqref{eq:DynReg1} and defining $C:=\textup{blkdiag}(\mathcal{Q}, \mathcal{R})$,
we can rewrite the dynamic regret problem as
\begin{equation}\label{eq:Reg2}
\begin{split}
\min_{\Phi}\max_{\norm{\mathbf{w}}^2=\omega} \bm{\delta}^\top \left( \Phi^\top C\Phi -O \right) \bm{\delta} \\
\textup{s.t. } \begin{bmatrix} \mathbb{I}-\mathcal{Z}\mathcal{A} & -\mathcal{Z}\mathcal{B} \end{bmatrix} \Phi = \mathcal{E}.
\end{split}
\end{equation}
Note that the inner maximisation in \eqref{eq:Reg2} is a non-convex quadratically constrained quadratic program (QCQP) and that the equality constraint on $\mathbf{w}$ can be relaxed to $\norm{\mathbf{w}}^2 \leq\omega$. In order to show that an exact solution to this problem can be found, we first make use of the fact that for QCQPs with one single quadratic constraint, strong duality always holds if Slater's condition is fulfilled, even if the objective or the constraint is non-convex. This result follows from the S-procedure, an implication of the commonly known S-lemma in optimisation theory, and detailed discussions can be found in \cite[Appendix B]{boyd2004convex}, \cite{polik2007survey} and \cite[Chapter 4.10.5]{ben2001lectures}.
In order to derive the dual to \eqref{eq:Reg2}, which allows formulating the main result in this letter, we first provide a quadratic form for \eqref{eq:Reg2}. By splitting $O$ into the block matrices  
\begin{equation*}
O= \left[\begin{array}{c : c} O_1 & O_2^T \\  \hdashline O_2 & O_3 \end{array}\right],
\end{equation*}
with $O_1\in\mathbb{R}^{n \times n}$, $O_2\in\mathbb{R}^{nT \times n}$ and $O_3\in\mathbb{R}^{nT \times nT}$, it follows that the inner maximisation in \eqref{eq:Reg2} is equivalent to
\begin{gather} \label{eq:Regquad}
\begin{align}
\max_{\norm{\mathbf{w}}^2\leq \omega} &\mathbf{w}^\top ( \Phi^{w\top} C \Phi^w {-}O_3 ) \mathbf{w} + 2x_0^\top ( \Phi^{0\top} C \Phi^w {-} O_2^\top ) \mathbf{w} \nonumber \\
&+ x_0^\top \left(\Phi^{0\top}C\Phi^{0}- O_1\right)x_0  
\end{align}
\end{gather}
This QCQP formulation leads to the following main result.

\begin{thrm} \label{thrm:ROSLS}
The solution to the dynamic regret problem \eqref{eq:DynReg1} can be obtained by solving
\begin{subequations}\label{eq:ROSLS}
\begin{align}
\gamma^*=\min_{\Phi,\gamma,\lambda} \;&\gamma \\
\textup{s.t. } &\lambda\geq 0 \\
& \begin{bmatrix} \mathbb{I}-\mathcal{Z}\mathcal{A} & -\mathcal{Z}\mathcal{B} \end{bmatrix}\Phi = \mathcal{E} \label{eq:SLPconstr} \\
&\hspace{-0.5cm} \begin{bmatrix}
x_0^\top O_1x_0{-}\lambda\omega{+}\gamma & x_0^\top O_2^\top &  x_0^\top\Phi^{0\top} \\ 
 O_2 x_0 & O_3+\lambda\mathbb{I} & \Phi^{w\top} \\ 
 \Phi^0x_0 & \Phi^w & C^{-1}
\end{bmatrix}\succeq 0 \label{eq:ROLMI} 
\end{align}
\end{subequations}
which is linear in the optimisation variables $\Phi, \lambda$ and $\gamma$. The regret optimal controller is given by $\mathcal{K}^*=\Phi_u^*\Phi_x^{*-1}$, where $\Phi_u^*$ and $\Phi_x^*$ are retrieved from an optimal solution $\Phi^*=[\Phi^{0*} \; \Phi^{w*}]$ of \eqref{eq:ROSLS} as defined in \eqref{eq:SLPrc}. The dynamic regret incurred by system \eqref{eq:sysdyn} is given by
\begin{equation} \label{eq:RegB}
\textup{Regret}\leq\textup{Regret}^* = \gamma^*.
\end{equation}
\end{thrm}
\begin{proof}
As the inner maximisation in \eqref{eq:Regquad} is strictly feasible using $\mathbf{w}=\bm{0}$, we can arrive at the strongly dual SDP formulation, as shown in e.g. \cite[Appendix B]{boyd2004convex},
\begin{subequations}
\begin{align}
\min_{\Phi,\gamma, \lambda} \;&\gamma \\
\textup{s.t. } &\lambda\geq 0 \\
&\hspace{-0cm}\begin{bmatrix} \mathbb{I}-\mathcal{Z}\mathcal{A} & -\mathcal{Z}\mathcal{B} \end{bmatrix}\Phi = \mathcal{E} \\
&\hspace{-0.8cm}\begin{bmatrix}
x_0^\top(O_1{-}\Phi^{0\top}C\Phi^0)x_0{-}\lambda\omega{+}\gamma  & x_0^\top O_2^\top{-}x_0^\top\Phi^{0\top}C\Phi^w \\ O_2x_0{-}\Phi^{w\top}C\Phi^0x_0 & \lambda\mathbb{I}{+}O_3{-}\Phi^{w\top}C\Phi^w
\end{bmatrix}\succeq0 \label{eq:dualLMI} 
\end{align}
\end{subequations} 
To arrive at formulation \eqref{eq:ROSLS}, we decompose \eqref{eq:dualLMI} into
\begin{equation*}
\!\begin{bmatrix}
x_0^\top O_1x_0{-}\lambda\omega{+}\gamma & x_0^\top O_2^\top \\ O_2x_0 & O_3{+}\lambda\mathbb{I}
\end{bmatrix}
-\begin{bmatrix}
x_0^\top\Phi^{0\top} \\ \Phi^{w\top}
\end{bmatrix}
\!C\!
\begin{bmatrix}
\Phi^0x_0 &  \Phi^w
\end{bmatrix} \succeq 0,
\end{equation*}
where the Schur complement can be applied, see e.g. \cite{vanantwerp2000tutorial}, as $C\succ0$ which follows from $\mathcal{Q}\succ 0$ and $\mathcal{R}\succ0$. This results in \eqref{eq:ROLMI}, which is linear in the optimisation variables. Through strong duality and the equivalence of LMIs using Schur complement, it holds that $\textup{Regret}\leq\gamma^*$ for any disturbance sequence such that $\norm{\mathbf{w}}^2\leq \omega$. The regret optimal control matrix $\mathcal{K}^*$ follows from the SLP in Section~\ref{sec:SLP}. 
\end{proof}
Note that the solution to Theorem~\ref{thrm:ROSLS} is not guaranteed to be unique nor is it guaranteed to be stable as a finite task horizon is considered. Furthermore, the provided bound \eqref{eq:RegB} is tight, i.e. there exists a disturbance sequence $\mathbf{w}^*$ with $\norm{\mathbf{w}^*}^2=\omega$ for any solution of \eqref{eq:ROSLS} such that using $\mathbf{u}=\mathcal{K}^*\mathbf{x}$ results in
\begin{equation}
\textup{Regret} = \textup{Regret}^*=\gamma^*,
\end{equation}
which follows directly from strong duality of \eqref{eq:DynReg1} and \eqref{eq:ROSLS}.
In practice, however, the maximal admissible disturbance energy is often unknown or can only be determined conservatively, i.e. $\omega$ is not available. Therefore, we provide a regret bound which only depends on the disturbance energy which the system actually experiences.

\begin{corollary}\label{cor:bound}
The dynamic regret achieved through the control sequence $\mathbf{u}=\mathcal{K}^*\mathbf{x}=\Phi_u^*\Phi_x^{*-1}\mathbf{x}$ resulting from \eqref{eq:ROSLS} in Theorem~\ref{thrm:ROSLS} for any given disturbance sequence and initial state is upper bounded by
\begin{equation}\label{eq:RegBDist}
\textup{Regret}\leq \sigma_{\textup{max}}(\Phi^{*\top}C\Phi^*-O)(\norm{x_0}^2+\norm{\mathbf{w}}^2).
\end{equation}
\end{corollary}
\begin{proof}
The regret system \eqref{eq:sysdyn} experiences under the control law $\mathbf{u}=\Phi_u^*\Phi_x^{*-1}\mathbf{x}$ is given by $\textup{Regret}=\bm{\delta}^\top(\Phi^{*\top}C\Phi^*-O)\bm{\delta}$ for any initial state and disturbance sequence. We can therefore show that
\begin{equation*}
\begin{split}
\textup{Regret}&=\norm{\bm{\delta}^\top(\Phi^{*\top}C\Phi^*-O)\bm{\delta}} \\
&= \norm{(\Phi^{*\top}C\Phi^*-O)^{\frac{1}{2}}\bm{\delta}}^2 \\
&\leq \norm{(\Phi^{*\top}C\Phi^*-O)}\norm{\bm{\delta}}^2\!\!,
\end{split}
\end{equation*}
which follows from the induced matrix norm definition and results in the provided bound as $\displaystyle\norm{\bm{\delta}}^2=\norm{x_0}^2+\norm{\mathbf{w}}^2$.
\end{proof}
The bound provided in Corollary~\ref{cor:bound} is again tight, in the sense that there exist $x_0^*$ and $\mathbf{w}^*$, such that $\textup{Regret}=\sigma_{\textup{max}}(\Phi^{*\top}C\Phi^*-O)(\norm{x_0^*}^2+\norm{\mathbf{w^*}}^2)$. Note that $$\sigma_{\textup{min}}(\Phi^{*\top}C\Phi^*-O) \leq \frac{\gamma^*}{\omega+\norm{x_0}^2}\leq \sigma_{\textup{max}}(\Phi^{*\top}C\Phi^*-O)$$ as the computation of $\gamma^*$ in \eqref{eq:ROSLS} makes use of a known initial condition with $\norm{\bm{\delta}}^2=\omega+\norm{x_0}^2$, which yields a different incurred regret than another initial condition with the same norm. If $x_0=0$, then the solution of \eqref{eq:ROSLS} achieves exactly the regret bound in \cite{goel2021regret}, i.e. $\textup{Regret}\leq\gamma^*\norm{w}^2$. If the initial condition is unknown or adversarially chosen, then $x_0$ can be used as an optimisation variable in the regret minimisation, i.e.
\begin{equation*}
\begin{split}
\min_{\Phi}& \max_{\norm{\bm{\delta}}^2=1} \bm{\delta}^\top (\Phi^\top C\Phi -O)\bm{\delta}  \\
\textup{s.t. } & \begin{bmatrix} \mathbb{I}-\mathcal{Z}\mathcal{A} & -\mathcal{Z}\mathcal{B} \end{bmatrix}\Phi = \mathcal{E},
\end{split}
\end{equation*}
which is equivalent to 
\begin{equation}\label{eq:ROSLSnorm}
\begin{split}
\min_{\Phi} & \norm{ \Phi^\top C\Phi -O}^2\\
\textup{s.t. } & \begin{bmatrix} \mathbb{I}-\mathcal{Z}\mathcal{A} & -\mathcal{Z}\mathcal{B} \end{bmatrix}\Phi = \mathcal{E},
\end{split}
\end{equation}
and is again solvable as an SDP linear in $\Phi$ through an epigraph transformation and using the fact that $\norm{ \Phi^\top C\Phi {-}O}^2{\leq} \gamma$ is equivalent to $ \Phi^\top C\Phi -O\preceq \gamma\mathbb{I}$, see e.g. \cite{vanantwerp2000tutorial}, for which the Schur complement can be used to express the constraint as an LMI, as similarly done in the proof of Theorem~\ref{thrm:ROSLS}. At this point, we would like to draw attention to the fact that similarly as to how the regret optimal problem is transformed into an $\mathcal{H}_\infty$-problem in \cite{goel2021regret}, the optimisation problem \eqref{eq:ROSLSnorm} bears similarities to the $\mathcal{H}_\infty$-synthesis using the SLP given in \cite{anderson2019system}, where the induced $2$-norm $\norm{\Phi^\top C\Phi}$ is minimised. 

\begin{remark}
It is often desirable to obtain regret bounds which are sublinear in the time horizon $T$ of the task as it implies that the benchmark policy is learned. While the regret bound \eqref{eq:RegBDist} in Corollary~\ref{cor:bound} depends on $\norm{\mathbf{w}}^2$, which clearly depends on the time horizon for most physical systems, the maximum singular value of $(\Phi^{*\top}C\Phi^*-O)$ is also dependent on $T$. If system \eqref{eq:sysdyn} is LTI, i.e. has constant system matrices, and the incurred stage costs also have constant $Q$ and $R$, then a time-independent constant $\tau\geq\sigma_{\textup{max}}(\Phi^{*\top}C\Phi^*-O)$ can be found $\forall T$ through the solution of the infinite horizon dynamic regret problem in \cite{sabag2021regret}. 
Albeit more conservative, the regret bound $\tau(\norm{x_0}^2+\norm{\mathbf{w}}^2)$ is dependent on $T$ only through the disturbance energy $\norm{\mathbf{w}}^2$. Sublinearity with respect to $T$ depends therefore on whether a sublinear bound on the disturbance energy with respect to $T$ exists. A disturbance which is pointwise bounded in time with the same bound as considered in Section~\ref{sec:PWB} achieves a regret bound $\mathcal{O}(T)$. Note that this is common in dynamic regret bounds in control, e.g. in \cite{nonhoff2021online} the dynamic regret bound is only sublinear in $T$ if the optimal pathlength is sublinear in $T$, as the best possible control sequence is considered as a benchmark instead of the restricted policy classes considered in policy regret.
\end{remark}
\section{STRUCTURED DYNAMIC REGRET PROBLEMS} \label{sec:Ext}
The SDP formulation which allows recovering the regret optimal controller in \eqref{eq:ROSLS} readily enables us to incorporate further structure of the disturbance signal other than the $\ell_2$-signal assumption considered in \cite{goel2021regret} and in \eqref{eq:DynReg1} as shown in Section~\ref{sec:PWB}. In Section~\ref{sec:Constr}, the ability to generate controllers such that state and input constraints on the system will be satisfied during an episode is presented.
\subsection{Pointwise bounded disturbances}\label{sec:PWB}
In this section, we show how ellipsoidal bounds on the disturbance can be leveraged in the dynamic regret synthesis problem and provide upper bounds on the dynamic regret incurred as well as lower bounds on the optimal solution.
\begin{assumpt}\label{ass:pwbell}
The disturbance $w_k$ in \eqref{eq:sysdyn} lies in a known, compact ellipsoid at every time step, i.e. 
$$w_k\in\mathcal{W}=\{w\in\mathbb{R}^n | \; w^\top P w\leq 1\},\ \forall k=0,\dots,T-1, $$
with $P \in\mathbb{R}^{n\times n}, P\succ0$.
\end{assumpt}
\noindent The resulting dynamic regret problem is given by
\begin{equation}\label{eq:RegPWB} 
\vspace{-0.1cm}
\textup{Regret}_{\textup{PWB}}^* = \min_\mathbf{u} \max_{\mathbf{w}\in \mathcal{W}^{T}} (J(x_0,\mathbf{u},\mathbf{w}) - J^*(x_0,\mathbf{w})). 
\end{equation}
The energy bound $\omega$ in \eqref{eq:DynReg1} is chosen such that $\norm{\mathbf{w}}^2\leq \omega$ if $\mathbf{w}\in\mathcal{W}^T$ in the remainder of this section. Note that the inner maximisation in \eqref{eq:RegPWB} is now a QCQP with $T$ constraints, and strong duality does not necessarily hold for this problem, see e.g. \cite{polik2007survey} and \cite{zheng2012zero} for discussions and counterexamples. Nonetheless, this problem can be approximated via an SDP formulation with regret no worse than \eqref{eq:DynReg1}. Before stating the main theorem of this section, we first present a result from \cite[Theorem 2]{ye1999approximating}, as restated in \cite[Proposition 4.10.5]{ben2001lectures}.
\begin{lemma}\label{lem:QCQPDual}
Consider the optimisation problem 
$$\bar{p}^* = \max_{x} x^\top D_0 x \textup{ s.t. } x^\top D_i x \leq c_i \; \forall i= 1,\dots, q$$
then the optimality ratio to the solution $\bar{d}^*$ of its dual SDP is given by 
$$  \bar{p}^* \geq \frac{2}{\pi} \bar{d}^*, $$
if the following conditions hold
\begin{enumerate}
\item The matrices $D_1, \dots, D_q$ commute with each other,
\item Slater's condition holds and there exists a combination of the matrices $D_1, \dots, D_q$ with nonnegative coefficients which is positive definite,
\item $D_0\succeq 0$.
\end{enumerate}
\end{lemma}
\begin{proof}
The proof is detailed in \cite[Proposition 4.10.5]{ben2001lectures} and \cite[Theorem 2]{ye1999approximating}.
\end{proof}
By showing that the conditions in Lemma~\ref{lem:QCQPDual} are fulfilled, we can provide a lower bound on the optimal regret for adversarial disturbances \eqref{eq:RegPWB} through the solution of the SDP dual problem.
\begin{thrm} \label{th:RegPWBell}
Under Assumption~\ref{ass:pwbell}, the solution of
\begin{gather}\label{eq:ROSLSPWB}
\begin{align}
\bar{\gamma}^*=\min_{\bar{\Phi},\bar{\lambda}} \;&\sum_{i=0}^T\bar{\lambda}_i \nonumber \\
\textup{s.t. } &\bar{\lambda}_i\geq 0 \; \forall i=0,\dots, T \nonumber \\
& \begin{bmatrix}  \mathbb{I}-\mathcal{Z}\mathcal{A} & -\mathcal{Z}\mathcal{B} \end{bmatrix}\bar{\Phi} = \mathcal{E},  \\
&\hspace{-1.0cm} \begin{bmatrix}
x_0^\top O_1x_0{+}\bar{\lambda}_T & x_0^\top O_2^\top &  x_0^\top\bar{\Phi}^{0\top} \\ 
O_2 x_0 & O_3{+}\sum\limits_{i=0}^{T-1}\bar{\lambda}_i\mathcal{P}_{i+1} & \bar{\Phi}^{w\top} \\ 
\bar{\Phi}^0x_0 & \bar{\Phi}^w & C^{-1}
\end{bmatrix}\succeq 0,  \nonumber
\end{align}
\end{gather}
which is linear in the optimisation variables $\bar{\Phi}$ and $\bar{\lambda}_i$ and where $[[\mathcal{P}_i]]_{i,i}=P$ with $[[\mathcal{P}_i]]_{j,l}=0$ for $(j,l)\neq (i,i)$, bounds the optimal dynamic regret incurred by system \eqref{eq:sysdyn}, under the control sequence $\mathbf{u}=\bar{\Phi}_u^*\bar{\Phi}_x^{*-1}\mathbf{x}$, as follows
\begin{equation} \label{eq:RegB1}
 \frac{2}{\pi} \bar{\gamma}^*\leq \textup{Regret}_{\textup{PWB}}^* \leq \bar{\gamma}^* \leq\gamma^*,
\end{equation}
where $\gamma^*$ denotes the solution of \eqref{eq:ROSLS}.
\end{thrm}
\begin{proof}
We first prove the lower bound on $\textup{Regret}_{\textup{PWB}}^*$, by showing that Lemma~\ref{lem:QCQPDual} holds since \eqref{eq:ROSLSPWB} is the dual of \eqref{eq:RegPWB}. The inhomogenous quadratic cost in \eqref{eq:RegPWB} can be rewritten as a homogenous quadratic form in 
$\mathbf{z}=[ \alpha , \mathbf{w}^\top ]^\top$ with $\alpha\in\mathbb{R}$ and the additional constraint $\mathbf{z}^\top \mathcal{Z} \mathbf{z} \leq 1$ with $[\mathcal{Z}]_{1,1}=1$ and $0$ in all other entries, as done similarly in \cite{ye1999approximating}. This reformulation introduces an additional row and column in the LMI condition with $\bar{\gamma}-\sum_{i=0}^T\bar{\lambda}_i$ on the diagonal and $0$ elsewhere, such that $\bar{\gamma}^*=\sum_{i=0}^T\bar{\lambda}^*_i$ is the optimal solution. Note that as $\mathcal{P}_i$ are block diagonal matrices with $P$ on the $i$-th entry and $0$ in all other entries, they commute with each other as their product is always $0$. Additionally, Slater's condition holds and $\sum_{i=1}^{T}\mathcal{P}_i\succ 0$ as $P\succ0$. Finally, $\bar{\Phi}^\top C\bar{\Phi}-O\succeq0$ for all systems satisfying \eqref{eq:SLPConstraint} under a causal control law which incurs a cost bigger or equal than the optimal non-causal cost. 
The bound $ \textup{Regret}_{\textup{PWB}}^* \leq \bar{\gamma}^*$ follows directly from Lagrangian weak duality. Lastly, as without loss of generality we assume that $\norm{\mathbf{w}}^2\leq \omega$ if $\mathbf{w}\in\mathcal{W}^{T}$, it holds that $w_k^\top w_k\leq \frac{\omega}{T}$ for all $w_k\in\mathcal{W}$ and therefore $\frac{\omega}{T}P\succeq \mathbb{I}$, which implies that $\frac{\omega}{T}\sum_{i=1}^{T}\mathcal{P}_i \succeq \mathbb{I}$. By setting $\bar{\lambda}_i= \frac{\lambda^*\omega}{T}$ for $i=0,\dots, T-1$, $\bar{\lambda}_T=\gamma^*-\lambda^*\omega$ and $\bar{\Phi}=\Phi^*$, \eqref{eq:ROSLSPWB} is feasible for $\bar{\gamma}=\gamma^*$ and therefore $\bar{\gamma}^*\leq\gamma^*$.
\end{proof}
Note that similarly, an upper bound on the optimal dynamic regret for adversarial disturbances can also be achieved if the disturbance is pointwise bounded in polytopes, however, Lemma~\ref{lem:QCQPDual} can no longer be applied. 
\subsection{State and input constrained systems}\label{sec:Constr}
In this section, we show how state and input constraints of the form
\begin{equation*}
\begin{split}
\mathcal{X}=\{x\in\mathbb{R}^n\;|\; H_x x\leq \mathbf{1}\}, \; \mathcal{U}=\{u\in\mathbb{R}^m \;|\; H_u u\leq \mathbf{1} \}
\end{split}
\end{equation*}
can be considered in the dynamic regret problem. The constraints are reformulated in the SLP as $H{_z}\Phi\bm{\delta}\leq\mathbf{1}$, $\forall \mathbf{w}\in\mathcal{W}^{T}$ with $H_z = \textup{blkdiag}(\mathbb{I}\otimes H_x, \mathbb{I}\otimes H_u)$, $H_z\in\mathbb{R}^{n_H(T+1){\times} (n{+}m)(T+1)}$,
which can be integrated in \eqref{eq:ROSLSPWB}. Under Assumption~\ref{ass:pwbell}, the state and input constraints can be enforced by considering the worst-case $\mathbf{w}$ exploiting the definition of the dual norm as shown, e.g. in \cite{goulart2006optimization}, as follows $\forall i=1,\dots, n_H$
\begin{gather} \label{eq:rodualconstr1}
\begin{align}
    &\max_{\mathbf{w} \in \mathcal{W}^{T}} [H_z]_{i,:} \Phi \bm{\delta} \leq 1, \nonumber \\
    \Leftarrow & [H_z]_{i,:}\Phi^0x_0+ {\sum_{j=1}^{T}}\norm{[H_{z}]_{i,:} [[\Phi^{w}]]_{:,j}  P^{-\frac{1}{2}}}\leq 1,
\end{align}
\end{gather}
which are second-order conic constraints and can simply be added to \eqref{eq:ROSLSPWB} in order to satisfy the constraints.
\begin{remark}
The constraints \eqref{eq:rodualconstr1} can be added to the synthesis problem \eqref{eq:ROSLSPWB} and an upper bound on the incurred dynamic regret is obtained. The synthesis of the optimal non-causal safe controller and the corresponding cost is discussed in \cite{martin2022safe}.
\end{remark}
\section{NUMERICAL EXAMPLE}
We consider the illustrative example of a discrete-time mass-spring-damper system with the dynamics
\begin{equation}
x_{k+1}=\begin{bmatrix} 1 & T_s \\ -cT_s & 1-dT_s \end{bmatrix}x_k+ \begin{bmatrix} 0 \\ T_s \end{bmatrix}u_k +w_k,
\end{equation}
with spring constant $c=0.2$, damping constant $d=0.1$ and sampling time $T_s=0.1$. The system is controlled over one episode with task length $T=100$ from the initial condition $x_0=[1\; 10]^\top$. At every time step, the system incurs the quadratic stage cost $l(x_k,u_k)=x_k^\top 0.1\mathbb{I}x_k+u_k^\top u_k$. The additive disturbance $w_k$, which the system experiences, is pointwise bounded in the ellipsoid 
\begin{equation}\label{eq:ExDistBound}
w_k\in\mathcal{W}=\{w\in\mathbb{R}^2 \mid w_k^\top w_k\leq 1\} \; \forall k=0,\dots, 99,
\end{equation}
such that the maximal admissible disturbance energy is bounded by $\norm{\mathbf{w}}^2=\mathbf{w}^\top\mathbf{w}\leq \omega = 100$. 
\begin{figure}[thpb]
\vspace{-0.4cm}
   \centering
   \hspace*{-0.53cm}\includegraphics[width=1.15\linewidth]{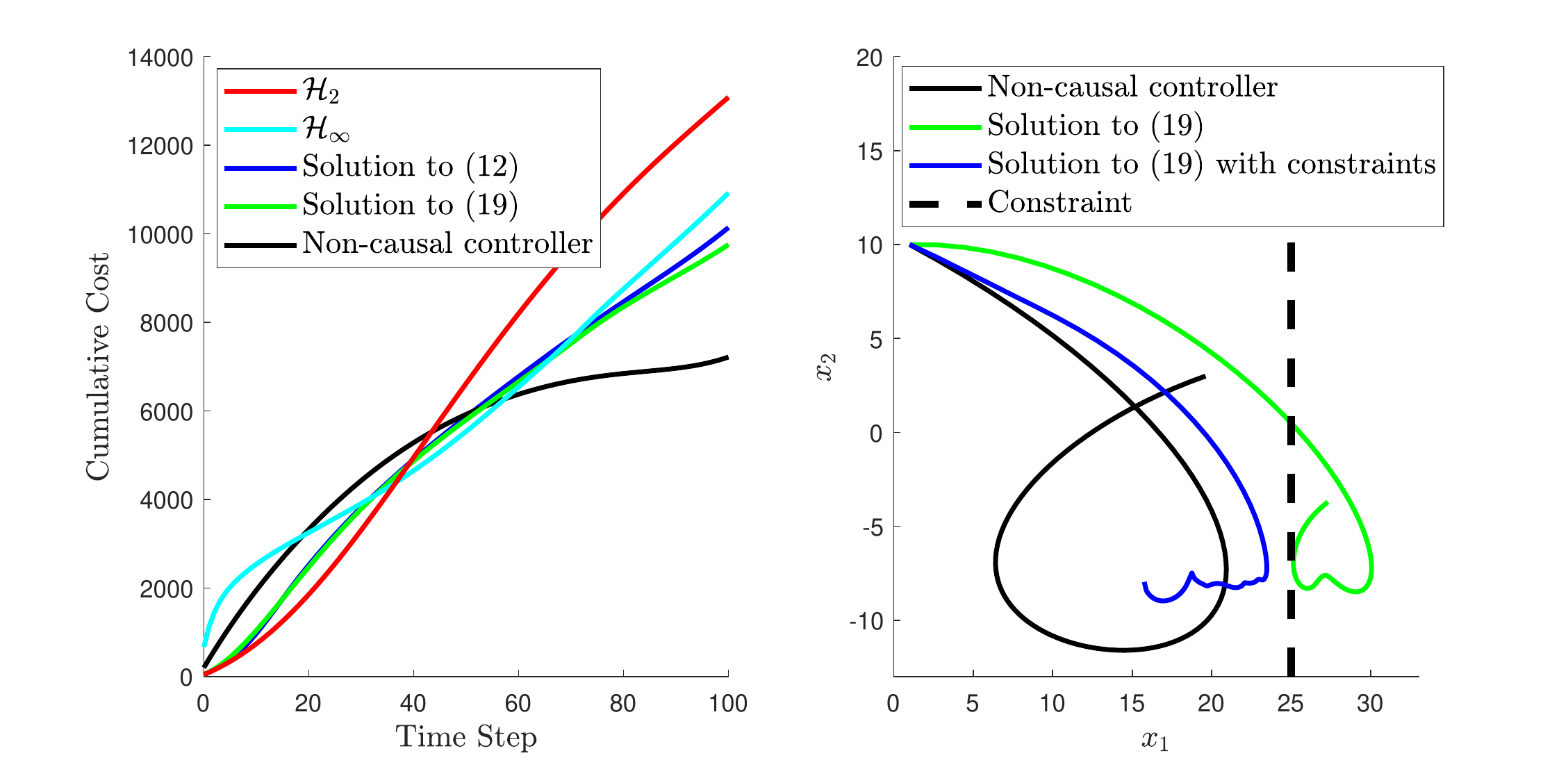} 
   \caption{Simulation of a mass-spring-damper system. \textit{Left:} The cumulative costs for one episode with horizon length $T=100$ are shown for the $\mathcal{H}_2$- and $\mathcal{H}_\infty$-controllers, the solution to the dynamic regret problem for bounded energy disturbances \eqref{eq:ROSLS}, the solution to the dynamic regret problem \eqref{eq:ROSLSPWB} with pointwise bounded energy and the optimal non-causal controller. \textit{Right:} The states arising from the optimal non-causal controller, the solution to \eqref{eq:ROSLSPWB} and the solution to \eqref{eq:ROSLSPWB} subject to the constraints \eqref{eq:ExConstr} are shown.}
   \label{fig:Sim}
\vspace{-0.5cm}
\end{figure}
\begin{table}[thpb]
\setlength{\tabcolsep}{4.7pt}
\caption{Total incurred costs and dynamic regrets}
\label{tbl:1}
\vspace{-0.4cm}
\begin{center}
\begin{tabular}{|c|c|c|c|c|c|}
\hline
Controller & $\mathcal{H}_2$ & $\mathcal{H}_\infty$ & $\eqref{eq:ROSLS}$ & $\eqref{eq:ROSLSPWB}$ & Non-causal 
\\
\hline
Incurred Cost & 13068 & 10925 & 10142 & 9755 & 7218 \\\hline
Incurred Regret & 5868  & 3707 & 2924 & 2537 & 0 \\
\hline
Regret Bounds & $\diagup$  & $\diagup$ & 4178 & 2955 & 0 \\\hline
\end{tabular}
\end{center}
\vspace{-0.2cm}
\end{table}
In Figure~\ref{fig:Sim}, the cumulative cost under the disturbance signal $\mathbf{w}=\frac{1}{\sqrt{2}}\mathbf{1}$ is shown for the optimal non-causal, $\mathcal{H}_2$- and $\mathcal{H}_\infty$-controllers, as well as the solutions to the dynamic regret problems \eqref{eq:ROSLS} and \eqref{eq:ROSLSPWB} which were solved using MOSEK \cite{mosek} and YALMIP \cite{lofberg2004yalmip}. Note that the $\mathcal{H}_2$-controller incurs the highest cost for the experienced disturbance signal and the dynamic regret optimal controller for bounded energy disturbances shows a better performance than the $\mathcal{H}_2$- and $\mathcal{H}_\infty$-controllers for the experienced disturbances. By solving the problem for the known pointwise ellipsoidal bounds \eqref{eq:ExDistBound} in \eqref{eq:ROSLSPWB}, the incurred cost is reduced by $4\%$ and the dynamic regret bound is lowered by $29\%$ from $4178$ to $2955$. The incurred cumulative costs and regrets can be seen in Table~\ref{tbl:1}. 

For the same system and disturbance signal, we now additionally consider the state and input constraints 
\begin{equation}\label{eq:ExConstr}
\mathcal{X}=\{x\in\mathbb{R}^2 \mid x_1 \leq 25\}, \; \mathcal{U}=\{u\in\mathbb{R} \mid \abs{u} \leq 15\},
\end{equation}
which are included in the form of \eqref{eq:rodualconstr1} in \eqref{eq:ROSLSPWB} and are thus guaranteed to be satisfied for all disturbances satisfying \eqref{eq:ExDistBound}. The resulting state trajectory for the constrained solution can be seen in Figure~\ref{fig:Sim} along with the trajectory resulting from the unconstrained solution of \eqref{eq:ROSLSPWB}, which violates the state constraints, and the optimal unconstrained non-causal controller. 
\section{CONCLUSION}
In this letter, we presented an SDP formulation based on the SLP, which allows to optimally solve the dynamic regret problem for bounded energy disturbances, as opposed to the bisection method presented in \cite{goel2021regret}, and achieves a dynamic regret of $\mathcal{O}(\norm{\mathbf{w}}^2)$. The novel SDP formulation enables the synthesis of controllers for structured dynamic regret problems. By utilising known pointwise ellipsoidal bounds on the disturbance, a no worse dynamic regret bound is incurred compared to using the disturbance energy and a lower bound on the optimal solution is provided. Additionally, the formulation allows directly incorporating state and input constraints on the system such that constraint satisfaction is guaranteed for the synthesised controller. Finally, we show that the proposed method can outperform the $\mathcal{H}_2$- and $\mathcal{H}_\infty$ controllers on a mass-spring-damper system. Potential future research includes extending the synthesis to the infinite horizon problem and to computationally more efficient approximations.
\bibliographystyle{IEEEtran} 
\bibliography{IEEEabrv,bibliography}
\end{document}